\def\mode{0}	% mode 0: regular paper
\if0\mode
\documentclass[journal, twoside, web]{ieeecolor}
\usepackage{generic}
\else
\documentclass[onecolumn, draftclsnofoot]{IEEEtran}
\fi

\let\labelindent\relax

\usepackage{xcolor}
\usepackage{amsmath,amssymb,amsthm}
\usepackage[short,c2,nocomma]{optidef}
\usepackage{nicefrac}
\usepackage{enumitem}
\usepackage[linesnumbered,ruled,vlined]{algorithm2e}
\usepackage{booktabs}
\usepackage{graphicx,epstopdf}
\usepackage{caption,subcaption}
\usepackage[nocompress]{cite}
\usepackage[colorlinks,linkcolor=blue]{hyperref}
\usepackage[nameinlink,capitalize]{cleveref}
\usepackage{orcidlink}

\newcommand{\bs}[1]{\boldsymbol{#1}}
\newcommand{\cl}[1]{\mathcal{#1}}
\newcommand{\bb}[1]{\mathbb{#1}}

\newcommand{\lc}{\left\{}
\newcommand{\rc}{\right\}}

\newcommand{\rank}[1]{\operatorname{rk}\lc#1\rc}
\newcommand{\supp}[1]{\operatorname{supp}\lc#1\rc}
\newcommand{\T}{\mathsf{T}}

\newtheorem{theorem}{Theorem}
\newtheorem{cor}[theorem]{Corollary}
\theoremstyle{definition}
\newtheorem{defn}{Definition}

\Crefname{cor}{Corollary}{Corollaries}

\newcommand{\revision}[1]{#1}

\title{Minimal Actuator Selection \\ \revision{for Linear Time-Invariant Systems}}

\author{Luca~Ballotta\textsuperscript{\orcidlink{0000-0002-6521-7142}}
	and~Geethu~Joseph\textsuperscript{\orcidlink{0000-0002-5289-5403}}
	\thanks{The authors are ordered by contribution.}%
	\thanks{Luca Ballotta is with the Department of Information Engineering, University of Padova, Padova 35131, Italy
		(e-mail: luca.ballotta@unipd.it).
		Work partially done at the Delft University of Technology.}%
	\thanks{Geethu Joseph is with the Signal Processing Systems Group, Delft University of Technology, 2628 CD Delft, The Netherlands
		(e-mail: g.joseph@tudelft.nl).}
}

\begin{document}
	
	\maketitle
	
	%!TEX root = MinimalActuatorSelection

\begin{abstract}
	
	Selecting a few available actuators to ensure the controllability of a linear system is a fundamental problem in control theory.
	Previous works either focus on optimal performance,
	simplifying the controllability issue,
	or make the system controllable under structural assumptions,
	such as in graphs or when the input matrix is a design parameter.
	We generalize these approaches to offer a precise characterization of the general \emph{minimal actuator selection} problem where a set of actuators is given,
	described by a fixed input matrix,
	and goal is to choose the fewest actuators that make the system controllable.
	We show that this problem can be equivalently cast as an integer linear program and,
	if actuation channels are sufficiently independent,
	as a set multicover problem under multiplicity constraints.
	The latter equivalence is always true if the state matrix has all distinct eigenvalues, in which case it simplifies to the set cover problem.
	Such characterizations hold even when a robust selection that tolerates a given number of faulty actuators is desired.
	Our established connection legitimates a designer to use algorithms from the rich literature on the set multicover problem to select the smallest subset of actuators,
	including exact solutions that do not require brute-force search.
	
\end{abstract}
	%!TEX root = ../MinimalActuatorSelection.tex

\section{Introduction}

Ensuring the controllability of a dynamical system is key to achieving regulation and desired behavior without degrading performance.
This challenge is exacerbated in modular and hierarchical systems built by interconnecting independent units,
which makes co-design tractable and robust~\cite{Zardini22cdc-modularCodesign,Incer25-Pacti}. 
As the size of the system grows, controllability becomes even more important
since strategically placing a few actuators can save costs and is practically feasible~\cite{DeLorenzo90jgcd-SensorActuatorSelection,Liu11nature-controllabilityNetworks}. So we investigate the problem of choosing the smallest subset of actuators from a given admissible set that renders the system controllable.

\subsubsection*{Related work}
\revision{Early work~\cite{VanDeWal01automatica-selectionIO} reviews actuator selection based on various requirements,
including controllability metrics.
That task has subsequently been tackled via optimization problems.
One main focus has been the reduction of control-theoretic costs in optimal control formulations.
For instance,
Zate~\emph{et al.} formulate a convex $\mathcal{H}_2$-optimal control problem using sparsity-promoting functions to reduce the number of actuators~\cite{Zare20tac-proximalAlgorithms},
Argha~\emph{et al.} codesign actuators and control law using linear matrix inequalities~\cite{Argha19ijc-optimalActuatorSensorSelection},
Tzoumas~\emph{et al.} study controller-actuator optimal codesign in LQG control~\cite{Tzoumas21tac-lqgCodesign},
and Manohar~\emph{et al.} leverage a balanced transformation of the Gramian matrices to select sensors and actuators~\cite{Manohar22tac-OptimalSensor}. 
A significant body of works focuses on networked and multi-agent control aiming to minimize a function of the controllability Gramian or LQR/LQG control cost~\cite{Guo21tac-actuatorPlacementGreedy,Pasqualetti14tcns-controllabilityMetricsNetworks,Clark12cdc-leaderSelection,Doostmohammadian20sj-driverNodes,Summers16tcns-submodularity,Baggio22ar-energyAwareControllability},\revision{\cite{Tzoumas2016tcns-minimalActuatorPlacement,Taha2017acc-actuatorSelectionCPS}}.
However,
network graphs induce specific structures on the state matrix and assume that each node can be individually actuated, imposing a diagonal input matrix.
A few works have developed time-varying actuator schedules to improve energy-related metrics~\cite{Siami21tac-actuatorScheduling}.
A recent line of work studies controllability when at most $s$ inputs can be nonzero at each time~\cite{Joseph24now-sparseActuatorControl},
offering an actuator schedule design with controllability guarantee in~\cite{Ballotta24lcss-actuatorSchedulingGuarantee}.
All these works focus on control-theoretic performance without openly discussing how to ensure controllability, with the exception of~\cite{Tzoumas2016tcns-minimalActuatorPlacement,Guo21tac-actuatorPlacementGreedy,Clark12cdc-leaderSelection}, which, however, assume a diagonal input matrix.}

\revision{A complementary line of research regards controllability as a design objective.
Most relevant to the current study is
the seminal paper~\cite{Olshevsky14tcns-minimalControllability}, which formulates the minimal controllability problem (MCP).
Here, the input matrix can be freely designed, and the objective is to minimize its nonzero (diagonal) elements, each representing a directly controlled state.
Pequito~\emph{et al.}~\cite{Pequito17automatica-robustMinimalControllability,Ramos21ijrnc-robustMinimalControllabilityObservability} extended the setup of~\cite{Olshevsky14tcns-minimalControllability} to robust selection, where a fixed number of faults can disable actuators,
establishing an equivalence with the set multicover problem and adapting a greedy selection algorithm.
However,
all elements of the input matrix are freely assigned, and its size is conservatively chosen to guarantee controllability under faults.
Zhang~\emph{et al.}~\cite{Zhang23tac-observabilityRobustnessSensorFailures} consider the dual problems of ensuring observability when sensors are removed and prove that it can be solved in polynomial time if the eigenspaces of the state matrix are bounded,
in striking contrast with the MCP.}

\subsubsection*{Contribution}
\revision{Although extensive research addresses actuator selection,
we identify two limitations.
First,
most works focus on control performance without explicitly addressing controllability.
Conversely, the MCP and its robust version aim to ensure controllability but assume that the input matrix can be freely designed.
This assumption works well for leader selection in multi-agent systems or control nodes in networks where each node corresponds to one state variable,
but it is not suited to the case when the input matrix is given,
corresponding to a set of predefined actuators.}

\revision{We fill this gap and study the \emph{minimal actuator selection} problem where we choose the minimal number of actuators that ensure controllability of a linear time-invariant (LTI) system under a given input matrix which encodes all and only available actuation channels.
First,
we show that this problem can be written as an integer linear program (ILP).
Second,
we prove that it is equivalent to the set multicover problem under a technical assumption which,
in words,
is satisfied if the available actuation channels are sufficiently independent with respect to the dynamics to be controlled.
The set multicover formulation reduces to the set cover problem if the state matrix has all distinct eigenvalues, which recovers the results in~\cite{Olshevsky14tcns-minimalControllability,Pequito17automatica-robustMinimalControllability}.
Also, we extend our formulation to the robust version, where a given number of actuators may fail,
proving that the ILP characterization and the set multicover equivalence require only modifying parameters of the ILP and of the technical assumption.
Finally, we review existing algorithms for the set multicover problem.
Numerical tests showcase the validity of the technical assumption for set multicover equivalence and compare the runtime and performance of ILP and set multicover algorithms from the literature.}

 \subsubsection*{Paper organization}
\revision{We formulate the minimal actuator selection problem,
 characterize it as an ILP,
 and discuss the equivalence with set multicover in \autoref{sec:minimal-actuator-selection}.
 We address the robust formulation in \autoref{sec:robust-minimal-actuator-selection}.
 We review and compare selection algorithms from control and computer science literature in \autoref{sec:algorithms},
 and report their comparison in terms of computational runtime and performance in \autoref{sec:experiments}
 along with a numerical validation of the condition required to establish the equivalence between actuator selection and the set multicover problem.
 We conclude and discuss the future outlook in \autoref{sec:conclusion}.}
	%!TEX root = ../MinimalActuatorSelection.tex

\section{Minimal Actuator Selection}
\label{sec:minimal-actuator-selection}

Consider the linear dynamical system $(\bs{A}\in\bb{R}^{n\times n},\bs{B}\in\bb{R}^{n\times m})$ with state $\bs{x}$,
input $\bs{u}$,
and state evolution at time $t$,
\begin{equation}
	\bs{x}(t+1) = \bs{A}\bs{x}(t)+\bs{B}\bs{u}(t) \mbox{ or } \dot{\bs{x}}(t) = \bs{A}\bs{x}(t)+\bs{B}\bs{u}(t). \label{eq:sysmodel}
\end{equation}
Each control input $u_i\in\bb{R}$, $i=1,\dots,m$, is associated with one actuator.
Since input $u_i$ affects system~\eqref{eq:sysmodel} through the $i$th column of matrix $\bs{B}$,
selection of actuators is equivalent to selecting columns of $\bs{B}$.
Our goal is to find a subset of actuators as small as possible,
denoted by $\cl{S}^*\subseteq[m]$,
such that the resulting system $(\bs{A}\in\bb{R}^{n\times n},\bs{B}_{\cl{S}^*}\in\bb{R}^{n\times |\cl{S}^*|})$ is controllable,
where $\bs{B}_{\cl{S}}$ represents the matrix composed by the columns of $\bs{B}$ with indices in the set $\cl{S}$.
This can be formalized as the optimization program
\begin{argmini}
	{\substack{\cl{S}\subseteq[m]}}
	{|\cl{S}|}
	{\label{eq:actuator_selection_intro}}
	{\cl{S}^*\in}
	\addConstraint{(\bs{A}, \bs{B}_{\cl{S}})}{\textnormal{ is controllable.}}{\space}
\end{argmini}
The system is controllable if and only if it satisfies the Popov-Belevitch-Hautus (PBH) test, or equivalently,
\begin{equation}\label{eq:PHB_minimal}
	\rank{\begin{bmatrix}\bs{A} - \lambda \bs{I} &  \bs{B}_{\cl{S}}\end{bmatrix}} = n \quad \forall\lambda\in\sigma(\bs{A})
\end{equation}
where $\sigma(\bs{A})$ denotes the spectrum of $\bs{A}$ without repeated eigenvalues.
We denote the number of distinct eigenvalues by $p\doteq|\sigma(\bs{A})|$.
The actuator selection~\eqref{eq:actuator_selection_intro} can be rewritten as
\begin{argmini}
	{\substack{\cl{S}\subseteq[m]}}
	{|\cl{S}|}
	{\label{eq:actuator_selection}}
	{\cl{S}^*\in}
	\addConstraint{\hspace{-0.5cm}\rank{\begin{bmatrix}\bs{A} - \lambda \bs{I} &  \bs{B}_{\cl{S}}\end{bmatrix}}}{=n}{\ \forall\lambda\in\sigma(\bs{A}).}
\end{argmini}

Unlike prior work~\cite{Olshevsky14tcns-minimalControllability,Pequito17automatica-robustMinimalControllability}, we do not assume that $\bs{A}$ is a simple matrix\footnote{
	The matrix has distinct eigenvalues, each with algebraic multiplicity 1.
}.
Although the rank constraint is nonconvex,
it is possible to cast it to a set of linear inequality constraints via manipulations involving the Jordan decomposition of the PBH test matrix.
We next reformulate the optimization problem~\eqref{eq:actuator_selection} as a binary ILP and discuss its implications.
We outline the steps to derive the parameters of the integer linear program in \Cref{alg:BILP},
while \Cref{thm:reformulation} presents its formulation.

\begin{algorithm}[t]
	\caption{Parameters of integer linear program}
	\label{alg:BILP}
	\DontPrintSemicolon
	\KwIn{Linear system $(\bs{A}\in\bb{R}^{n\times n},\bs{B}\in\bb{R}^{n\times m})$}
	\KwOut{	Matrices $\bs{W}^{(i)}$ for $i\in[p]$}		
	Compute the Jordan decomposition $\bs{A}=\bs{P}\bs{J}\bs{P}^{-1}$\;
	Identify distinct eigenvalues $\{\lambda_i\}_{i\in[p]}$ of $\bs{A}$ with geometric multiplicities $\{g_i\}_{i\in[p]}$\;
	Compute $\bar{\bs{B}}=\bs{P}^{-1}\bs{B}$\;
	\For {$i=1,2,\ldots,p$}
	{
		Define $\cl{G}_i$ as the $g_i$ zero row indices of $\bs{J}-\lambda_i\bs{I}$\;
		Find all subsets $\cl{S}^{(k)}\subseteq[m]$, for $k\in[\alpha_i]$, such that
		$\rank{\bar{\bs{B}}_{\cl{G}_i,\cl{S}^{(k)}}}=g_i \;\;\text{and}\;\; |\cl{S}^{(k)}|=g_i.$\nllabel{line:rank_condition}\;
		Construct matrix $\bs{W}^{(i)}\in\{0,1\}^{\alpha_i\times m}$ such that $[\bs{W}^{(i)}]_{kj}=1 \iff j\in\cl{S}^{(k)}.$
	}
\end{algorithm}

\begin{theorem}\label{thm:reformulation}
	Consider the discrete-time linear dynamical system~\eqref{eq:sysmodel}.
	Let $\bs{W}^{(i)}\in\bb{R}^{\alpha_i\times m}$, $i\in[p]$, be obtained from \Cref{alg:BILP}.
	Then, the minimal actuator selection problem~\eqref{eq:actuator_selection} is equivalent to 
	\begin{argmini}
		{\substack{\bs{y}\in\{0,1\}^m\\\bs{d}^{(i)}\in\{0,1\}^{\alpha_i}}}
		{\bs{1}^{\T}\bs{y}}
		{\label{eq:minimal_prob_refor}}
		{\bs{y}^*\in}
		\addConstraint{\bs{W}^{(i)}\bs{y}}{\geq \left(\bs{W}^{(i)}\bs{1}\right)\odot\bs{d}^{(i)}}{\ \forall i\in[p]}
		\addConstraint{\bs{1}^{\top}\bs{d}^{(i)}}{\geq 1}{\ \forall i\in[p].}
	\end{argmini}
	where $\cl{S}^*=\supp{\bs{y}^*}$ and $\odot$ is the Hadamard product.
\end{theorem}
\begin{proof}
	See \Cref{app:reformulation}.
\end{proof}
\Cref{thm:reformulation} establishes that minimal actuator selection can be posed as an integer linear program. 
Parameter $\bs{W}^{(i)}$ is a selection matrix whose rows encode all and only minimal-cardinality subsets of columns of $\bs{B}$ that are non-orthogonal to the (left) eigenspace of matrix $\bs{A}$ corresponding to eigenvalue $\lambda_i$.
In words,
those columns of $\bs{B}$ select the actuators in charge of controlling the dynamics associated with $\lambda_i$.
Since its eigenspace has dimension $g_i$,
each row of $\bs{W}^{(i)}$ selects $g_i$ actuators to satisfy the constraints.
Moreover,
selecting one such a subset to $\cl{S}$ for each distinct eigenvalue of $\bs{A}$,
which is expressed by the constraints in program~\eqref{eq:minimal_prob_refor}, is necessary and sufficient to satisfy the PBH test as \eqref{eq:PHB_minimal} is equivalent to
\begin{equation}\label{eq:PHB_minimal_refor}
	\bs{v} \bs{B}_{\cl{S}} \neq 0 \quad \forall \bs{v}: \bs{v}\bs{A} = \lambda\bs{v} \quad \forall\lambda\in\sigma(\bs{A}).
\end{equation}
All steps of \Cref{alg:BILP} can be completed in polynomial time except for \autoref{line:rank_condition}.\footnote{
	\revision{We assume standard polynomial-complexity routines are used to compute eigenvalues up to machine precision.}
}
This step is inherently combinatorial because it finds all full-rank submatrices composed of $g_i$ columns of $\bar{\bs{B}}_{\cl{G}_i}$. 
In the worst case, this involves enumerating $\binom{m}{g_i}$ choices.
Therefore, the complexity of formulating problem~\eqref{eq:minimal_prob_refor} is polynomial in $m$ and $n$ for the class of systems such that $G(\bs{A})\doteq\max_i g_i$ is bounded by a universal constant, i.e., independent of all system parameters.

The linear problem formulation~\eqref{eq:minimal_prob_refor} can be approximately solved efficiently by dedicated solvers even for fairly large $n$ and $m$.
However,
in general, this does not give a precise indication of the computational complexity of solving it.
Next, we prove that the problem is NP-complete under a technical assumption on the system matrices, thereby establishing a formal equivalence with the set multicover problem.

\begin{defn}[Full spark frame~\cite{alexeev2012full}]
	A collection of $m$ vectors $\mathcal{V}=\{\bs{a}_i\}_{i=1}^m$ with $\bs{a}_i\in\bb{R}^{n}$ for all $i\in[m]$ is a \emph{full spark frame} if any $n$ vectors of $\mathcal{V}$ are linearly independent; equivalently
	\begin{equation}
		\rank{\begin{bmatrix}
				\bs{a}_{j_1} & \dots & \bs{a}_{j_n}
		\end{bmatrix}} = n \quad \forall \{j_1,\dots,j_n\}\subseteq [m].
	\end{equation}
\end{defn}

\revision{In words,
	$m\ge n$ vectors form a full spark frame if any $n$ of them are a basis of $\bb{R}^n$.}

\begin{theorem}\label{prop:setmulticover}
	Consider the discrete-time linear dynamical system~\eqref{eq:sysmodel} such that $G(\bs{A})\leq G$ where $G$ is a universal constant.
	Assume that, for each matrix $\bar{\bs{B}}_{\cl{G}_i}\in\bb{R}^{g_i\times m}$ computed by \Cref{alg:BILP}, there exists a set $\cl{T}_i\subset[m]$ such that $\bar{\bs{B}}_{\cl{G}_i,\cl{T}_i}$ is a full spark frame and $\bar{\bs{B}}_{\cl{G}_i,\cl{T}_i^{\complement}}=\bs{0}$.
	Then, the minimal actuator selection problem~\eqref{eq:actuator_selection} is NP-complete. 
\end{theorem}
\begin{proof}
    \revision{The proof relies on the equivalence between~\eqref{eq:actuator_selection} and the set multicover problem.
	See details in \Cref{app:setmulticover}.}
\end{proof}

If a set $\cl{T}_i$ corresponding to a full spark frame such that $\bar{\bs{B}}_{\cl{G}_i,\cl{T}_i^\complement}=\bs{0}$ exists,
it is unique by definition.
Therefore,
verifying the second assumption of \Cref{prop:setmulticover} requires polynomial time in $m$ given that $G(\bs{A})$ is universally bounded,
as one needs to remove all zero-columns of $\bar{\bs{B}}_{\cl{G}_i}$ and evaluate all combinations the remaining columns with cardinality $g_i$.
It follows that checking the full spark frame assumption does not compromise the NP-completeness of the whole procedure.

The full spark frame assumption may seem fairly strong.
In words,
it requires that all actuators affecting the same eigenvalue $\lambda_i$ (i.e., such that their associated columns in $\bs{B}$ are non-orthogonal to eigenvectors of $\lambda_i$) are sufficiently independent,
such that any $g_i$ of them can control the dynamics associated with $\lambda_i$.
\revision{This assumption formalizes the intuition that excessive actuation redundancy increases the complexity of the actuator selection problem. By reducing redundant choices, the selection problem becomes more structured and can be easier to solve than a generic ILP formulation, even when the underlying problem remains NP-complete.}
%This quantifies the intuition that actuation redundancy should be reduced to simplify the selection problem, since NP-complete programs \revision{have lower computational complexity} than generic ILPs.

Notably,
the full spark frame is an algebraic generalization of conditions used in previous works on minimal controllability,
as stated by the next ancillary result.
\begin{cor}\label{cor:setcover}
	Consider the discrete-time linear dynamical system~\eqref{eq:sysmodel}.
	If $g_i=1 \; \forall\lambda\in\sigma(\bs{A})$, then the minimal actuator selection problem~\eqref{eq:actuator_selection} is NP-complete.
\end{cor}
\begin{proof}
    \revision{If $g_i\equiv1$,
    	every $\bar{\bs{B}}_{\cl{G}_i,\cl{T}_i}$ is a full spark frame and problem~\eqref{eq:actuator_selection} reduces to minimal set cover; see \Cref{app:setcover}}. 
\end{proof}
A result similar to \Cref{cor:setcover} is proven in \cite{Pequito17automatica-robustMinimalControllability} for the special case when $\bs{A}$ is a simple matrix using a slightly different technique,
whereas~\cite{Olshevsky14tcns-minimalControllability} proves only NP-hardness.
In~\cite{Pequito17automatica-robustMinimalControllability},
the problem formulation allows the designer to freely choose the elements of matrix $\bs{B}$,
which is therefore chosen by horizontally concatenating diagonal matrices and pruning redundant columns.
Additionally,
in~\cite{Pequito17automatica-robustMinimalControllability} all eigenvalues of matrix $\bs{A}$ have unit algebraic, and hence geometric, multiplicities.
By contrast,
\Cref{prop:setmulticover} proves NP-completeness of the more general case where $\bs{A}$ has repeated eigenvalues which are possibly associated with high-dimensional eigenspaces.
	%!TEX root = ../MinimalActuatorSelection.tex

\section{Robust Minimal Actuator Selection}
\label{sec:robust-minimal-actuator-selection}

The previous section deals with the nominal selection task,
implicitly assuming that all actuators keep working correctly at all times.
This scenario is not robust to accidental or malicious faults, which may deactivate one or multiple actuators,
formally corresponding to removing some of the selected columns from matrix $\bs{B}_{\cl{S}}$.

To remedy this,
we consider the fault-aware version of problem~\eqref{eq:actuator_selection} where the selection is made robust to a number of simultaneous actuator faults.
Considering the scenario where any $f<m$ selected actuators may fail,
for some pre-defined robustness parameter $f$,
we formulate the robust minimal actuator selection as
\begin{argmini}
	{\substack{\cl{S}\subseteq[m]}}
	{|\cl{S}|}
	{\label{eq:robust_actuator_selection}}
	{\cl{S}_\textnormal{r}^*\in}
	\addConstraint{\rank{\begin{bmatrix}\bs{A} - \lambda \bs{I} &  \bs{B}_{\cl{S}\setminus\cl{F}}\end{bmatrix}}}{=n}{\null}
    \addConstraint{\hspace{4mm}\forall\lambda\in\sigma(\bs{A}),\forall\cl{F}\subset\cl{S}: |\cl{F}|\le f}.
\end{argmini}
This robust problem formulation is mathematically similar to the nominal selection~\eqref{eq:actuator_selection}.
The key difference is that possible faults require selecting redundant actuators to satisfy the rank condition of the PBH test~\eqref{eq:PHB_minimal} in all cases.
In particular,
the PBH test must hold valid for all eigenvalues according to the equivalence in~\eqref{eq:PHB_minimal_refor} even if $f$ actuators associated with the same eigenvalue $\lambda$ fail. We formalize this discussion with the following results.
\begin{theorem}\label{thm:robust-reformulation}
	Consider the discrete-time linear dynamical system~\eqref{eq:sysmodel}.
	Let $\bs{W}^{(i)}\in\bb{R}^{\beta_i\times m}$, $i\in[p]$, be obtained from \Cref{alg:BILP} after replacing \autoref{line:rank_condition} with the following:
	
	Find all subsets $\cl{S}^{(k)}\subseteq[m]$, for $k\in[\beta_i]$, such that
	\begin{equation}\label{eq:rank_condition_robust}
		\bar{\bs{B}}_{\cl{G}_i,\cl{S}^{(k)}} \ \text{is a full spark frame and} \ |\cl{S}^{(k)}|=\min\{g_i+f,m\}.
	\end{equation}
	Then, the robust minimal actuator selection problem~\eqref{eq:robust_actuator_selection} is equivalent to program~\eqref{eq:minimal_prob_refor},
	where $\cl{S}_{\textnormal{r}}^*=\supp{\bs{y}^*}$.
\end{theorem}
\begin{proof}
	See \Cref{app:robust-reformulation}.
\end{proof}

\begin{cor}\label{cor:robust-feasible}
	The robust minimal actuator selection problem~\eqref{eq:robust_actuator_selection} is feasible if and only if $g_i+f\le m$ for all $i\in[m]$.
\end{cor}
\begin{proof}
	It readily follows from condition~\eqref{eq:rank_condition_robust} and the constructive arguments in \Cref{app:robust-reformulation}.
\end{proof}

\Cref{thm:robust-reformulation} states that program~\eqref{eq:minimal_prob_refor} accommodates actuator faults with no significant modification.
The key difference lies in the parameter matrices $\bs{W}^{(i)}$, which are responsible for providing the actuator redundancy necessary to counterbalance faults.
In particular,
fault-aware selection requires a construction based on full spark frames,
which translates to several full-rank submatrices for each mode rather than just one full-rank submatrix.
Formally,
this is because any $f$ selected columns of $\bs{B}$ associated with the $i$th mode may be zeroed out,
requiring the remaining selection to fulfill the PBH test with respect to $\lambda_i$.
In words,
condition~\eqref{eq:rank_condition_robust} and \Cref{cor:robust-feasible} mean that each controlled mode needs $f$ additional actuators compared to the fault-free selection and that the selected actuators must be sufficiently independent.
Since \Cref{thm:robust-reformulation} establishes an equivalence between problems~\eqref{eq:robust_actuator_selection} and~\eqref{eq:minimal_prob_refor},
this is a necessary,
although possibly strong,
requirement.
In fact, no selection is robust if there are few actuation channels compared to the uncontrolled dynamics, as highlighted in \Cref{cor:robust-feasible}.

If we assume feasibility,
checking condition~\eqref{eq:rank_condition_robust} requires assessing that submatrices composed by $g_i+f$ columns of $\bs{\bar{\bs{B}}}_{\cl{G}_i}$ are full spark frames.
This computation has worst-case time complexity proportional to $\binom{m}{g_i+f}\binom{g_i+f}{g_i}$ for each eigenvalue $\lambda_i$.
Hence,
analogously to the discussion in \autoref{sec:minimal-actuator-selection},
the time complexity of reformulating problem~\eqref{eq:robust_actuator_selection} to an ILP is still polynomial in $m$ provided that $G(\bs{A})$ is universally bounded.
The equivalence to the set multicover problem is analogous to fault-free selection,
as well.
\begin{theorem}\label{prop:robust_setmulticover}
	Consider the discrete-time linear dynamical system~\eqref{eq:sysmodel} such that $G(\bs{A})\leq G$ where $G$ is a universal constant.
	Assume that, for each matrix $\bar{\bs{B}}_{\cl{G}_i}\in\bb{R}^{g_i\times m}$ computed by \Cref{alg:BILP},
	there exists a set $\cl{T}_i\subset[m]$ with $|\cl{T}_i|\ge g_i+f$ such that $\bar{\bs{B}}_{\cl{G}_i,\cl{T}_i}$ is a full spark frame and $\bar{\bs{B}}_{\cl{G}_i,\cl{T}_i^{\complement}}=\bs{0}$.
	Then, the robust minimal actuator selection problem~\eqref{eq:robust_actuator_selection} is NP-complete. 
\end{theorem}
\begin{proof}
	See \Cref{app:robust_setmulticover}.
\end{proof}

Similar to the ILP reformulation,
\Cref{thm:robust-reformulation} reveals that the only difference between considering and neglecting faults is that $f$ redundant actuators have to be selected to drive each dynamical mode,
with a suitable independence criterion given by the full spark frame.
	%!TEX root = ../MinimalActuatorSelection.tex

\section{Algorithms for Actuator Selection}
\label{sec:algorithms}

The ILP~\eqref{eq:minimal_prob_refor} can be processed by off-the-shelf solvers that provide good solutions in practice.
A popular strategy in the control literature is greedy selection of actuators to heuristically solve~\eqref{eq:actuator_selection_intro} or~\eqref{eq:actuator_selection} while ensuring polynomially bounded time complexity.
This property is attractive for large-scale systems,
whereas ILP solvers based on,
e.g.,
branch-and-bound strategies quickly grow in computational complexity and may incur unacceptable runtime for large or ill-conditioned problem instances.
The seminal work~\cite{Olshevsky14tcns-minimalControllability} proposes greedy algorithms for the special case where \revision{the numerical values} of $\bs{B}$ can be chosen,
reducing the problem to carefully selecting columns of a diagonal matrix.
Moreover,
variants of the greedy approach incorporate control-theoretic costs,
such as Gramian-based empirical metrics of control energy.
Such a strategy is used to heuristically provide controllability for the ``dual'' problem where an upper bound on selected actuators is imposed,
which often works well in practice for suitable choices of the control cost and well-conditioned system matrices~\cite{Summers16tcns-submodularity,Siami21tac-actuatorScheduling,Baggio22ar-energyAwareControllability}.
On the other hand,
a few greedy procedures with formal controllability guarantees have been proposed, e.g., by leveraging the matroid structure of feasible actuator sets in leader selection~\cite{Clark12cdc-leaderSelection,Guo21tac-actuatorPlacementGreedy} and finite-horizon controllability sets for $s$-sparse actuator scheduling~\cite{Ballotta24lcss-actuatorSchedulingGuarantee}.

\subsubsection*{Set multicover}
Stepping forward from the general problem~\eqref{eq:actuator_selection},
\cref{prop:setmulticover,prop:robust_setmulticover} provide us with added structure under technical assumptions.
The equivalence with the set multicover problem they establish unlocks the use of efficient algorithms which (i) ensure controllability and (ii) enjoy provable bounds on selected actuators, given that the set multicover problem possesses the matroid structure needed to quantify suboptimality of greedy selection.

The set cover problem is defined as follows.
Given a universe $\cl{N}$ of elements uniquely labeled as $\cl{N}=\{1,\dots,N\}$ and $S$ subsets $\cl{S}_i\subset\cl{N}$, $i\in[S]$,
find a smallest subset collection $\{\cl{S}_{j_i}\}_{i=1}^s\subseteq\{\cl{S}_i\}_{i=1}^S$ such that all elements are covered by (included in) the chosen collection, i.e., $\cl{N} = \cup_{i=1}^s\cl{S}_{j_i}$.
Set multicover is a natural extension whereby each element $i$ has to be covered (included) $b_i\ge1$ times in the chosen collection.
In our case,
the universe $\cl{N}$ corresponds to the distinct eigenvalues (or eigenspaces) of $\bs{A}$ with $N=p$,
subset $\cl{S}_i$, $i\in[m]$, gathers all eigenvalues ``covered'' by the $i$th actuator, namely, whose eigenvectors are non-orthogonal to the actuator according to~\eqref{eq:PHB_minimal_refor},
and the coverage requirement of each unique eigenvalue $\lambda_i$ is $b_i=g_i$.
The problem considered here is the version where each subset can be chosen at most a given number of times (so-called set multicover with multiplicity constraints),
which in our case is $1$ since selecting an actuator multiple times has no practical meaning.
This is formally represented through the constraint $y_i\le1$ in problem~\eqref{eq:minimal_prob_refor}.
The linear program formulations of set cover and multicover problems are provided in \Cref{app:setmulticover,app:setcover},
respectively.
If $|\cl{S}_i|\le k$ for all $i$ and a known $k\le N$,
then the problem is called $k$-set multicover,
which is solvable in polynomial time for $k=2$ and is NP-complete for $k\ge3$.
For the actuator selection problem,
this bound quantifies the largest number of eigenvectors associated with different eigenvalues that are non-orthogonal to any actuator and can be computed as $k=\max_{i\in[p]}|\bs{W}^{(i)}|_1$.

\subsubsection*{Covering algorithms}
Greedy selection performs empirically well for covering problems
and it has received a great deal of attention in operations research literature.
Reference~\cite{Pequito17automatica-robustMinimalControllability} adapts the greedy procedure for set multicover to minimal robust actuator selection in the special case where matrix $\bs{A}$ is simple with $p=n$ and $\bs{B}$ can be arbitrarily chosen with $m=(f+1)n$,
reporting the classic $O(\log n)$ gap on performance ratio.
On the other hand,
we can find essentially no relevant work in the control literature that uses other strategies than greedy selection,
excluding ad-hoc solutions.
\revision{This is typically justified by the combinatorial complexity of the original selection problem and the need for computationally tractable approaches.}
However,
it is noteworthy that other effective algorithms are available to solve covering problems, and they can be seamlessly adapted to the scenario at hand if the full spark frame assumptions hold.
In particular,
exact algorithms may be effective in practical scenarios at the cost of increased computation and/or memory requirements,
such as if the selection has to be performed once offline.

\begin{table}
	\centering
	\caption{Approximate algorithms for set cover and multicover.}
	\label{tab:algo-approx}
	\begin{tabular}{cccc}
		\toprule
		Problem			& Type			& Suboptimality gap					& Ref.\\
		\midrule
		Set cover		& Greedy		&$\log p - \log\log p + \Theta(1)$	& \cite{Slavik96stc-greedySetCoverAnalysis}\\
		Set cover		& Greedy+opt.	& $H(p) - \nicefrac{1}{2}$			& \cite{Duh97stoc-approximationkSetCover}\\
		Set multicover	& Greedy		& $H(p)$							& \cite{Kolliopoulos05jcss-approximationAlgorithmsIntegerPrograms}\\
		Set multicover	& Greedy+opt.	& $H(p) - \nicefrac{1}{6}$			& \cite{Fujito06springer-betterThanGreedySetMulticover}\\
		\bottomrule
	\end{tabular}
\end{table}

\begin{table}
	\centering
	\caption{Exact algorithms for set cover and multicover.}
	\label{tab:algo-exact}
	\begin{tabular}{cccc}
		\toprule
		Problem			& Time			& Memory		& Ref.\\
		\midrule
		Set cover		& $O^*(2^p)$	& $O^*(2^p)$	& \cite{Bjorklund09sjc-setPartitioning}\\
		Set cover		& $O^*(m2^p)$	& Polynomial	& \cite{Bjorklund09sjc-setPartitioning}\\
		Set multicover	& $O(m(G(\bs{A})+1)^p)$	& $O^*((G(\bs{A})+1)^p)$	& \cite{Hua10tcs-dynamicProgrammingSetMulticover}\\
		\bottomrule
	\end{tabular}
\end{table}

We collect a few approximate and exact algorithms in \Cref{tab:algo-approx,tab:algo-exact},
respectively,
which are adapted from~\cite{Hua10tcs-dynamicProgrammingSetMulticover}.
Symbol $H(p)$ in \autoref{tab:algo-approx} denotes the $p$th Harmonic number $H(p) = \sum_{c=1}^p c^{-1}$ for which it holds $\log(p+1)\le H(p)\le\log p+1$,
whereas notation $O^*(f(n))$ in \autoref{tab:algo-exact} omits a factor in $O(\log f(n))^{O(1)}$.
We report \revision{multiplicative} suboptimality gaps for approximate algorithms that enjoy polynomial time complexity,
and runtime and memory bounds for exact algorithms.\footnote{
	If algorithm $\cl{A}$ has suboptimality gap $\gamma$,
	then $S_\cl{A}\le \gamma S_\textnormal{min}$ where $S_\cl{A}$ is the number of actuators output by $\cl{A}$ and $S_\textnormal{min}$ is an optimal solution.
}
All bounds are reported with the notation used in the present paper.
Recall that the set cover problem corresponds to the case where all distinct eigenvalues of $\bs{A}$ have geometric multiplicity $g_i\equiv1$,
including the special case with $p=n$.

Paper~\cite{Slavik96stc-greedySetCoverAnalysis} proves an exact gap of the greedy algorithm for set cover,
namely its $\Theta$-asymptotic expression characterizes both lower and upper bounds which differ by at most $1.09$ independently of all parameters.
In particular,
the upper bound is more precise than the $O(\log p)$ gap independently proved in~\cite{Olshevsky14tcns-minimalControllability}.
Such a strong result is not known for the set multicover problem,
to the best of our knowledge.
Improved algorithms that provide better bounds by combining an initial greedy selection with a subsequent optimal search on a reduced problem,
which can be cast to maximum matching and thus exactly solvable in polynomial time,
have been presented in~\cite{Duh97stoc-approximationkSetCover,Fujito06springer-betterThanGreedySetMulticover} respectively for set cover and multicover.
As for the latter problem,
reference~\cite{Fujito06springer-betterThanGreedySetMulticover} proves that the bound is tight for the proposed algorithm.
Interestingly,
all bounds depend on the number of distinct eigenvalues $p\le n$ but not on the number of actuators $m$ or geometric multiplicities $g_i$.
For $k$-set (multi)cover problems,
parameter $p$ can actually be replaced with bound $k<p$ in the suboptimality gaps of \cref{tab:algo-approx}.

More recently,
algorithms that exactly solve set cover and multicover problems have been proposed.
We report them in \autoref{tab:algo-exact}, including the only reference we can find with an exact algorithm for set multicover with multiplicity constraints,
since most literature focuses on the unconstrained version.
The NP-complete nature of the problem is evident from the exponential complexity bounds.
Different from the suboptimality gaps of approximate algorithms,
such bounds depend both on $p$,
$g_i$, and $m$, which also affect the polynomial-time problem construction through \Cref{alg:BILP}.
The exact algorithms in \cref{tab:algo-exact} do not exhaustively evaluate all possible combinations of actuators,
which has time and memory complexity $O(2^m)$.
If $m$ is large, those approaches may be significantly cheaper from a computational perspective
and preferred whenever sufficient runtime and memory are available.
	%!TEX root = ../MinimalActuatorSelection.tex

\section{Numerical Experiments}
\label{sec:experiments}

\revision{We showcase the theoretical results and selection performance with three numerical tests.
We implement all algorithms in MATLAB and solve the ILP~\eqref{eq:minimal_prob_refor} with Gurobi.}

\subsubsection*{Academic example}
\revision{First, we slightly modify the discrete-time system in~\cite[Example~1]{Ballotta24lcss-actuatorSchedulingGuarantee},
obtaining matrices
\begin{equation}\label{eq:ex}
	\bs{A} = \begin{bmatrix}
		0 & 1 & 0 & 0 & 0\\
		0 & 0 & 0 & 1 & 0\\
		0 & 0 & 1 & 0 & 0 \\
		0 & 0 & 0 & 0 & 1 \\
		0 & 0 & 0 & 0 & 0
	\end{bmatrix} \quad
	\bs{B} = \begin{bmatrix}
		0 &	1 & 0 & 0 &	0\\
		0 &	0 &	0 & 1 &	0\\
		1 &	0 &	0 & 0 &	0\\
		1 &	0 &	0 & 0 &	1\\
		0 &	1 &	1 & 0 &	0
	\end{bmatrix}.
\end{equation}
Because matrix $\bs{A}$ has the repeated eigenvalue $\lambda_1=0$ with geometric multiplicity $g_1=2$ and matrix $\bs{B}$ is not diagonal,
the methods in~\cite{Olshevsky14tcns-minimalControllability,Tzoumas2016tcns-minimalActuatorPlacement} cannot be used.
The ILP~\eqref{eq:minimal_prob_refor} with \cref{alg:BILP} finds the minimal-cardinality actuator subset $\mathcal{S}^*=\{1,3\}$.
In this case, the full spark frame assumption holds, and \cref{prop:setmulticover} ensures equivalence with the set multicover problem.
Both set multicover greedy selection and the exact algorithm in~\cite{Hua10tcs-dynamicProgrammingSetMulticover} find minimal actuator subsets of size two.
}

\begin{figure}
	\centering
	\includegraphics[width=0.7\linewidth]{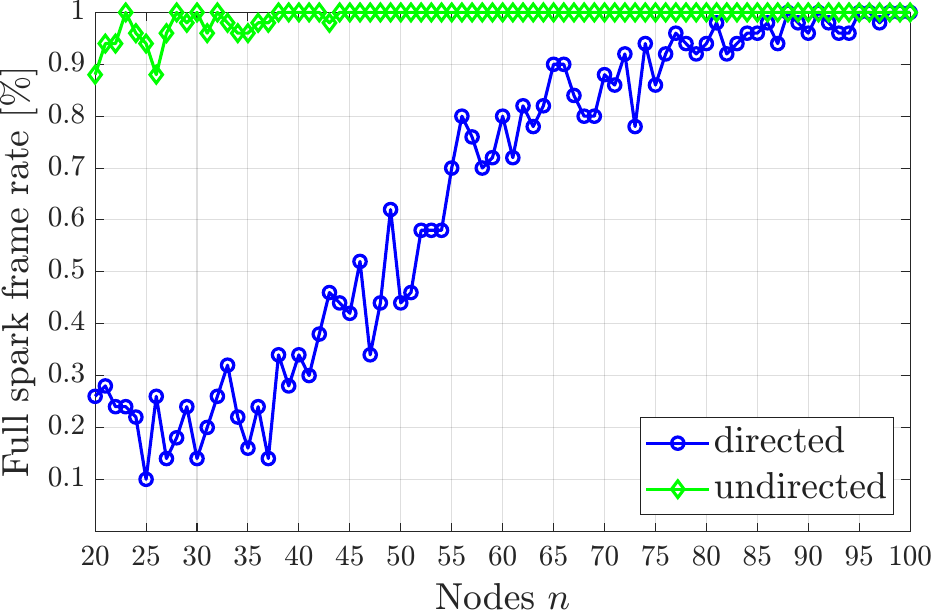}
	\caption{Satisfaction rate of full spark frame assumption in random systems where $\bs{A}$ is a distance-weighted adjacency matrix of a random geometric graph and $\bs{B}=\bs{I}$.}
	\label{fig:fullsparkrate}
\end{figure}

\subsubsection*{Full spark frame}
\revision{We seek to gain insight into the technical condition in \cref{prop:setmulticover,prop:robust_setmulticover} which enables the equivalence between actuator selection~\eqref{eq:actuator_selection} and~\eqref{eq:robust_actuator_selection} and set multicover problem.
To do so, we construct systems with increasing state dimension $n$ and $\bs{B}=\bs{I}$.
For each $n$, we generate $50$ system instances and compute the number of times the full spark frame assumption in \cref{prop:setmulticover,prop:robust_setmulticover} is satisfied.
We compute matrix $\bs{A}$ by generating random geometric graphs with $n$ nodes and radius $\rho=0.25$, whereby we set $\bs{A}_{ij} = \textnormal{e}^{-\lVert p_i - p_j\rVert}$ where $p_i\in[0,1]^2$ is the position of the $i$th node.
We set $\bs{A}_{ii}\equiv-0.1$ following~\cite{Baggio22ar-energyAwareControllability}.
Also, we build corresponding directed graphs by randomly assigning one-way to each edge.
\autoref{fig:fullsparkrate} shows that undirected graphs nearly always satisfy the assumption even under low connectivity (small $n$),
while directed graphs need to be densely connected.
This suggests that the sparser couplings in directed graphs fail the full spark frame condition more easily due to excessive actuation redundancy,
while dynamical symmetries make selection easier.
Notably,
this study is more general than graphs;
the full spark frame condition involves ranks of submatrices,
thus, nonzero patterns of $\bs{A}$ and $\bs{B}$ are more relevant than numerical values.
We acknowledge that choosing $\bs{B}=\bs{I}$ helps computational tractability but limits the scope of this experiment;
a deeper analytical and numerical study will be pursued in future work.}

\begin{figure}
	\begin{minipage}{.5\columnwidth}
		\centering
		\includegraphics[height=.95\linewidth]{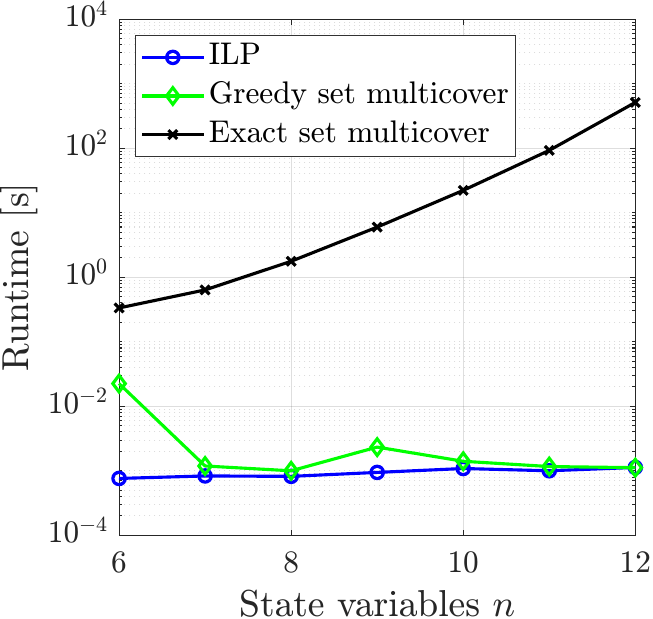}
		\caption{Runtime with increasing $n$.}
		\label{fig:graph5to12}
	\end{minipage}%
	\begin{minipage}{.5\columnwidth}
		\centering
		\includegraphics[height=.95\linewidth]{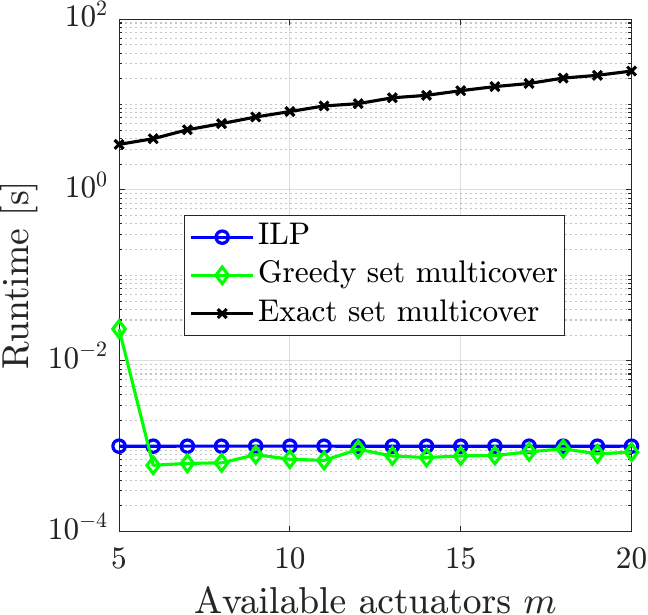}
		\caption{Runtime with increasing $m$.}
		\label{fig:graph10}
	\end{minipage}
\end{figure}

\begin{figure}
	\begin{subfigure}{0.5\columnwidth}
		\centering
		\includegraphics[height=.95\linewidth]{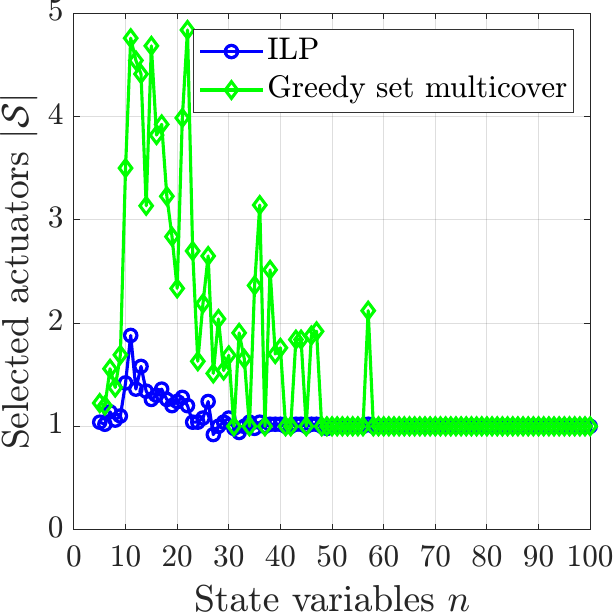}
		\caption{Average selected actuators.}
		\label{fig:cost}
	\end{subfigure}%
	\begin{subfigure}{0.5\columnwidth}
		\centering
		\includegraphics[height=.95\linewidth]{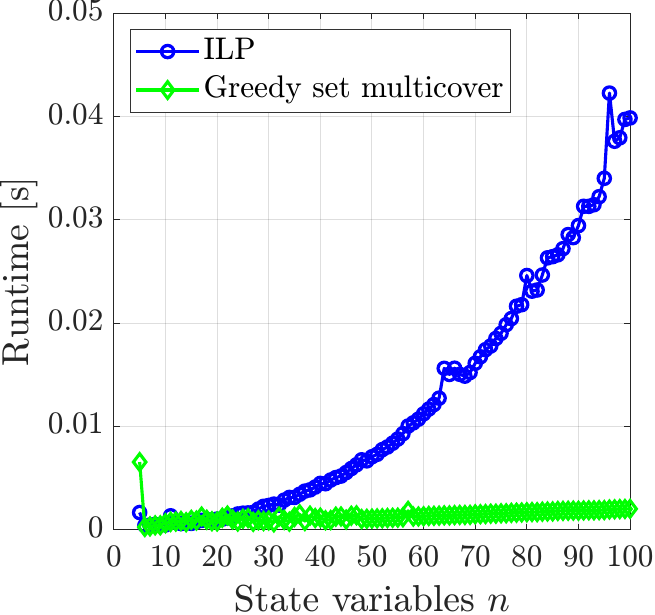}
		\caption{Runtime.}
		\label{fig:runtime}
	\end{subfigure}
	\caption{Actuator selection in random undirected network systems.}
	\label{fig:graph5to100}
\end{figure}

\subsubsection*{Network dynamics}
\revision{We now compare the ILP~\eqref{eq:minimal_prob_refor},
greedy selection for set multicover, and the exact algorithm in~\cite{Hua10tcs-dynamicProgrammingSetMulticover} on random network systems.
We generate matrix $\bs{A}$ as previously described, 
but now we generate $\bs{B}$ by randomly assigning one or two $1$'s to each column,
i.e., each actuator directly affects one or two state variable(s),
differently from the standard setup where $\bs{B}$ is assumed or made diagonal~\cite{Tzoumas2016tcns-minimalActuatorPlacement,Olshevsky14tcns-minimalControllability,Guo21tac-actuatorPlacementGreedy,Baggio22ar-energyAwareControllability}.
\cref{fig:graph5to12} compares the runtime as $n$ increases, with $m=n$,
while \cref{fig:graph10} shows it as $m$ increases with $n=10$.
In the latter test, matrix $\bs{B}$ is progressively enlarged so that,
for each value of $m$,
the set of available actuators includes all actuators available at $m-1$ plus the new $m$th column of $\bs{B}$.
The runtime of ILP and greedy is comparable while,
for the exact algorithm,
it grows exponentially with $n$ and polynomially with $m$,
consistently with the bound in \cref{tab:algo-exact}.
Finally, we compare ILP and greedy selection for larger values of $n$.
For each $n$, we average across $50$ random systems with $\rho=0.5$ if $n<10$ and $\rho=0.3$ otherwise, to ensure that the full spark frame assumption is satisfied sufficiently many times to implement greedy selection for set multicover.
The results are reported in \cref{fig:graph5to100}.
Greedy selects up to three times the number of actuators w.r.t. ILP for $n\le50$,
but it reaches the same minimal number of actuators (one) for $n\ge50$.
On the other hand, the runtime of ILP grows much faster,
although it remains practically competitive with the range of $n$ used.
}
	%!TEX root = ../MinimalActuatorSelection.tex

\section{Conclusion}
\label{sec:conclusion}

\revision{By leveraging the PBH test and Jordan decomposition,
we have reformulated the minimal actuator selection problem to a binary ILP.
If the actuators affecting each mode are sufficiently independent,
as precisely quantified by a linear-algebraic structure, the full spark frame,
minimal actuator selection simplifies to the set multicover problem,
and reduces to set cover if all eigenspaces of $\bs{A}$ are one-dimensional.
We have formulated a robust version where $f$ selected actuators may be removed and shown that the same formal equivalences hold by suitably modifying parameters of the ILP reformulation.
This study advances understanding of actuator selection problems and strengthens the bridge between control-theoretic resource allocation and combinatorial problems in operations research.}
Compelling directions for future studies include \revision{a deeper characterization of the full spark frame assumption},
time-varying actuator schedules,
and trading a minimal actuator set for optimal performance or minimal (average) control energy.
	
	\appendices
	\crefalias{section}{appendix}
	
	%!TEX root = ../MinimalActuatorSelection.tex

\section{Proof of \Cref{thm:reformulation}}\label{app:reformulation}
Let $\bs{A}=\bs{P}\bs{J}\bs{P}^{-1}$ be the real Jordan normal form of $\bs{A}$ where the upper triangular matrix $\bs{J}$ is a Jordan matrix. Then, \eqref{eq:PHB_minimal} is equivalent to
\begin{align}
	n &= \rank{\begin{bmatrix}\bs{P}\bs{J}\bs{P}^{-1} - \lambda \bs{I} &  \bs{B}_{\cl{S}}\end{bmatrix}} \\
	&= \rank{\begin{bmatrix}\bs{J} - \lambda \bs{I} &  \bs{P}^{-1}\bs{B}_{\cl{S}}\end{bmatrix}\begin{bmatrix}
			\bs{P}& \bs{0}\\
			\bs{0} &\bs{I}
	\end{bmatrix}}\\
	&= \rank{\begin{bmatrix}\bs{J} - \lambda \bs{I} & \bar{\bs{B}}_{\cl{S}}\end{bmatrix}},\label{eq:PHB_modified}
\end{align}
where $\bar{\bs{B}}=\bs{P}^{-1}\bs{B}$.  Further, let the distinct eigenvalues of $\bs{A}$ be $\lambda_1,\lambda_2,\ldots,\lambda_p$, with algebraic and geometric multiplicities $(a_1,g_1),(a_2,g_2),\ldots,(a_p,g_p)$, respectively. Then, we have
\begin{equation}
	\bs{J} - \lambda \bs{I} = \mathrm{blkdiag}\{\bs{J}^{(1)} - \lambda\bs{I}, \ldots,\bs{J}^{(p)} - \lambda\bs{I}\}.
\end{equation}
Here, $\bs{J}^{(i)}\in\bb{R}^{a_i\times a_i}$ contains the $g_i$ Jordan blocks corresponding to the eigenvalue $\lambda_i$.
Hence, matrix $\bs{J} - \lambda \bs{I}$ is rank deficient only when $\lambda=\lambda_i$ for some $i$ as $\bs{J}^{(i)}-\lambda_i\bs{I}$ has exactly $g_i$ zero rows (the bottom row of each Jordan block) and its rank is $n-g_i$.
Let $\cl{G}_i$ be the indices of the zero rows of $\bs{J}-\lambda_i\bs{I}$.
Hence, \eqref{eq:PHB_minimal} holds if and only if for any $\lambda=\lambda_i$
\begin{align}
	n&= \rank{\begin{bmatrix}[\bs{J} - \lambda \bs{I}]_{\cl{G}_i^{\complement}} & \bar{\bs{B}}_{{\cl{G}_i^{\complement}},\cl{S}}\\
			\bs{0} & \bar{\bs{B}}_{{\cl{G}_i},\cl{S}}\end{bmatrix}}\\
	& = (n-g_i)+\rank{\bar{\bs{B}}_{{\cl{G}_i},\cl{S}}}.
\end{align}
As a result, the minimum controllability problem can be formulated as 
\begin{argmini}
	{\substack{\cl{S}\subseteq[m]}}
	{|\cl{S}|}
	{\label{eq:minimal_prob}}
	{\cl{S}^*\in}
	\addConstraint{\rank{\bar{\bs{B}}_{\cl{G}_i,\cl{S}}}}{=g_i}{\ \forall i\in[p].}
\end{argmini}
To satisfy the condition $\rank{\bar{\bs{B}}_{\cl{G}_i,\cl{S}}}=g_i$, we need to select at least $g_i$ columns of $\bar{\bs{B}}_{\cl{G}_i}\in\bb{R}^{g_i\times n}$ to $\cl{S}$. 
Let $\alpha_i$ be the number of subsets $\cl{S}_i$ representing $g_i$ columns of $\bar{\bs{B}}_{\cl{G}_i}$ that satisfy the rank condition,
\begin{equation}\label{eq:ithrank}
	\rank{\bar{\bs{B}}_{\cl{G}_i,\cl{S}_i}}=g_i.
\end{equation}
Then, $\bs{W}^{(i)}\in \{0,1\}^{\alpha_i\times m}$ in \Cref{alg:BILP} encodes the mapping between such subsets and the columns of $\bar{\bs{B}}_{\cl{G}_i}$. Hence, it has $g_i$ nonzero entries per row and we have 
\begin{equation}\label{eq:prop_W}
	\bs{W}^{(i)}\bs{1}=g_i\bs{1}.
\end{equation}
Let $\bs{y}$ be a selection vector that indicates if a column of $\bs{B}$ is included in $\cl{S}$ or not. Then, $\bs{W}^{(i)}\bs{y}\in\bb{R}^{\alpha_i}$ indicates the number of columns from each of the subsets $\cl{S}_i$'s are selected to $\cl{S}$. Since we need at least one of the feasible $\cl{S}_i\subset\cl{S}$ to satisfy the rank condition \eqref{eq:ithrank}, we rewrite \eqref{eq:ithrank} as
\begin{equation}\label{eq:max_con}
	\max\;\bs{W}^{(i)}\bs{y}=g_i.
\end{equation}
Also, from \eqref{eq:prop_W}, we deduce that $\bs{0}\leq \bs{W}^{(i)}\bs{y} \leq  g_i\bs{1}$. Therefore, $\bs{W}^{(i)}\bs{y} \geq g_i\bs{d}^{(i)}$, where the binary slack variable $\bs{d}^{(i)}\in\{0,1\}^{\alpha_i}$ whose $j$th entry indicates whether the $j$th entry of $\bs{W}^{(i)}\bs{y}$ is $g_i$. So, we reformulate \eqref{eq:max_con} as
\begin{equation}\label{eq:constraint}
	\bs{W}^{(i)}\bs{y} \geq g_i\bs{d}^{(i)}= (\bs{W}^{(i)}\bs{1})\odot\bs{d}^{(i)} \ \text{and} \ \bs{1}^{\top}\bs{d}^{(i)}\geq 1.
\end{equation}
Further, noting that $\bs{1}^{\T}\bs{y}=|\cl{S}|$, we express the optimization problem \eqref{eq:minimal_prob} as \eqref{eq:minimal_prob_refor} with $\cl{S}^*=\supp{\bs{y}^*}$. Thus, the proof is complete.
	%!TEX root = ../MinimalActuatorSelection.tex

\section{Proof of \Cref{prop:setmulticover}}\label{app:setmulticover}
Reference~\cite{Olshevsky14tcns-minimalControllability} shows that the minimum set cover problem can be polynomially reduced to the actuator selection problem, and therefore, the problem is NP-hard.
To prove NP-completeness, we show that our problem can be reduced to the minimum set multicover problem in polynomial time. 

We recall from \Cref{app:reformulation} that the minimum controllability problem can be formulated as \eqref{eq:minimal_prob}. Since $\bar{\bs{B}}_{\cl{G}_i,\cl{T}_i^{\complement}}=\bs{0}$, we note that any subset $\cl{S}_i$ satisfying \eqref{eq:ithrank} is a subset of $\cl{T}_i$. Also, because $\bar{\bs{B}}_{\cl{G}_i,\cl{T}_i}$ is a full spark frame, any subset of $\cl{T}_i$ with cardinality $g_i$ satisfies \eqref{eq:ithrank}. Therefore, problem~\eqref{eq:minimal_prob} is equivalent to
\begin{argmini}
	{\substack{\cl{S}\subseteq[m]}}
	{|\cl{S}|}
	{\label{eq:minimal_prob1}}
	{\cl{S}^*\in}
	\addConstraint{|\cl{T}_i\cap\cl{S}|}{\ge g_i}{\ \forall i\in[p].}
\end{argmini}
Now, we define $m$ sets $\cl{R}_j=\{i\in[p]: j\in\cl{T}_i\}$,
with $j=1,\dots,m$,
such that 
\begin{equation}
	\cl{T}_i\cap\cl{S} = \{j\in\cl{S}:\;i\in\cl{R}_j\}.
\end{equation}
In words,
the set $\cl{R}_j$ collects the eigenvalues ``covered'' by actuator $j$,
meaning that their corresponding eigenvectors are non-orthogonal to the $j$th column of $\bs{B}$.
This allows us to rewrite problem~\eqref{eq:minimal_prob1} as
\begin{argmini}
	{\substack{\cl{S}\subseteq[m]}}
	{|\cl{S}|}
	{\label{eq:minimal_prob2}}
	{\cl{S}^*\in}
	\addConstraint{|\{j\in\cl{S}:\;i\in\cl{R}_j\}|}{\ge g_i}{\ \forall i\in[p]}
\end{argmini}
which is a set multicover problem.
Since the reduction from the minimal actuator selection problem to~\eqref{eq:minimal_prob2} is polynomial time when $G(\bs{A})$ is bounded, the proof is complete.

\subsubsection*{ILP formulation}
We can also reduce the linear program~\eqref{eq:minimal_prob_refor} to minimum set multicover problem. For this, we note that $\bar{\bs{B}}_{\cl{G}_i,\cl{T}_i^{\complement}}=\bs{0}$ implies that no row in $\bs{W}^{(i)}$ contains indices in $\cl{T}_i^{\complement}$. 
Also, the assumption that $\bar{\bs{B}}_{\cl{G}_i,\cl{T}_i}$ is a full spark frame implies that, for any $g_i$ indices in $\cl{T}_i$,
there exists a row in $\bs{W}^{(i)}$ where the entries are $1$ in the positions corresponding to those selected indices and $0$ elsewhere.
So, if $\bs{y}_{\cl{T}_i}\in\{0,1\}^{\alpha_i\times t_i}$ has $g_i$ nonzero entries, then at least one entry of $\bs{W}^{(i)}\bs{y}$ is equal to $g_i$. Hence, the constraint in~\eqref{eq:constraint} reduces to $\sum_{j\in \cl{T}_i}\bs{y}_{j}=g_i$. As a result, problem~\eqref{eq:minimal_prob_refor} becomes
\begin{argmini}
    {\substack{\bs{y}\in\{0,1\}^m}}
    {\bs{1}^{\T}\bs{y}}
    {}
    {\bs{y}^*\in}
    \addConstraint{\sum_{j\in \cl{T}_i}\bs{y}_{j}}{\ge g_i}{\ \forall i\in[p]},
\end{argmini}
which is an ILP version of the set multicover problem.
	%!TEX root = ../MinimalActuatorSelection.tex

\section{Proof of \Cref{cor:setcover}}\label{app:setcover}
Since the geometric multiplicity is one for all eigenvalues, the matrix $\bar{\bs{B}}_{\cl{G}_i}$ obtained from \Cref{alg:BILP} is a row vector for $i\in[p]$. Since any row vector without zero entries is a full spark frame, the conditions of \Cref{prop:setmulticover} are satisfied, and the problem is NP-complete.

\subsubsection*{ILP formulation}
In this case, $\bs{W}^{(i)}\in\{0,1\}^{\alpha_i\times m}$ is a submatrix of $\bs{I}$ formed using the $\alpha_i$ rows indexed by the nonzero entries of $\bar{\bs{B}}_{\cl{G}_i}$. Therefore, the binary slack variable $\bs{d}^{(i)}=\bs{W}^{(i)}\bs{y}\in\{0,1\}^{\alpha_i}$ satisfies the constraint in \eqref{eq:minimal_prob} and it reduces to
\begin{argmini}
    {\substack{\bs{y}\in\{0,1\}^m}}
    {\bs{1}^{\T}\bs{y}}
    {\label{eq:minimal_reduction}}
    {\bs{y}^*\in}
    \addConstraint{\bs{1}^{\top}\bs{W}^{(i)}\bs{y}}{\geq 1}{\ \forall i\in[p]}.
\end{argmini}
Here, $\bs{1}^{\top}\bs{W}^{(i)}\in\{0,1\}^{1\times m}$ and the formulation is the linear program formulation of the set cover problem.
	%!TEX root = ../MinimalActuatorSelection.tex

\section{Proof of \Cref{thm:robust-reformulation}}\label{app:robust-reformulation}

The proof is analogous to the one in \Cref{app:reformulation},
with the difference that all subsets of actuators that are non-orthogonal to the left eigenvectors of $\bs{A}$ associated with each distinct eigenvalue $\lambda\in\sigma(\bs{A})$ must remain such after removing any $f'\le f$ columns of $\bs{B}$.
Since up to $f$ actuators may fail and rank is non-decreasing with the number of columns selected from $\bs{B}$,
the worst-case scenario corresponds to removing exactly $f$ columns of $\bs{B}_{\cl{S}}$ associated with the same eigenvalue $\lambda$ according to~\eqref{eq:PHB_minimal_refor}.
Note that every column of $\bar{\bs{B}}$ computed by \Cref{alg:BILP} is zero if and only if the corresponding column of $\bs{B}$ is zero.
Thus,
problem~\eqref{eq:robust_actuator_selection} is equivalent to
\begin{argmini}
	{\substack{\cl{S}\subseteq[m]}}
	{|\cl{S}|}
	{\label{eq:robust_minimal_prob}}
	{\cl{S}_\textnormal{r}^*\in}
	\addConstraint{\rank{\bar{\bs{B}}_{\cl{G}_i,\cl{S}\setminus\cl{F}}}}{=g_i}{\;\forall i\in[p],\forall\cl{F}\subset\cl{S}}
	\addConstraint{|\cl{F}|}{= f.}
\end{argmini}
Since matrix $\bar{\bs{B}}_{\cl{G}_i}$ has $g_i$ rows and rank is sub-additive,
it holds
\begin{equation}\label{eq:rank-requirement-robust}
	\rank{\bar{\bs{B}}_{\cl{G}_i,\cl{S}\setminus\cl{F}}}\le\min\{g_i,|\cl{S}|-|\cl{F}|\}\le\min\{g_i,|\cl{S}|-f\}.
\end{equation}
It follows that at least $g_i+f$ columns of $\bar{\bs{B}}_{\cl{G}_i}$ must be selected to $\cl{S}$ to satisfy the first constraint in~\eqref{eq:robust_minimal_prob}.
Condition~\eqref{eq:ithrank} is consequently modified as
\begin{equation}\label{eq:robust_ithrank}
	\rank{\bar{\bs{B}}_{\cl{G}_i,\cl{C}_i}} = g_i \ \forall \cl{C}_i\subset\cl{S}_i: |\cl{C}_i| = g_i,
\end{equation}
where $\cl{S}_i$ is a subset of $\min\{g_i+f,m\}$ columns of $\bar{\bs{B}}_{\cl{G}_i}$.
Equivalently,
$\bar{\bs{B}}_{\cl{G}_i,\cl{S}_i}$ is a full spark frame.
This guarantees that,
if $m\ge g_i+f$ and any $f$ elements of $\cl{S}_i$ are removed,
the remaining column indices in $\cl{C}_i$ satisfy the PBH test rank condition~\eqref{eq:PHB_minimal_refor} for eigenvalue $\lambda_i$.
Encoding these new subsets $\cl{S}_i$ into matrix $\bs{W}^{(i)}$ yields
\begin{equation}\label{eq:robust_prop_W}
	\bs{W}^{(i)}\bs{1}=\min\{g_i+f,m\}\bs{1},
\end{equation}
From~\eqref{eq:rank-requirement-robust},
the requirement~\eqref{eq:max_con} becomes
\begin{equation}\label{eq:robust_max_con}
	\max\;\bs{W}^{(i)}\bs{y}=g_i+f.
\end{equation}
Finally,
in analogy to~\eqref{eq:constraint},
the slack variable $\bs{d}^{(i)}$ serves to select at least one row of $\bs{W}^{(i)}$,
imposing the constraints
\begin{align}\label{eq:robust_constraint}
	\bs{W}^{(i)}\bs{y} &\geq\min\{g_i+f,m\}\bs{d}^{(i)}= (\bs{W}^{(i)}\bs{1})\odot\bs{d}^{(i)} \\
	\bs{1}^{\top}\bs{d}^{(i)}&\geq 1.
\end{align}
Since the two constraints above appear in ILP~\eqref{eq:minimal_prob_refor},
the proof is complete.
	%!TEX root = ../MinimalActuatorSelection.tex

\section{Proof of \Cref{prop:robust_setmulticover}}\label{app:robust_setmulticover}

Analogously to the nominal minimal actuator selection,
we note that $\bar{\bs{B}}_{\cl{G}_i,\cl{T}_i^\complement}=\bs{0}$ and that any subset $\cl{S}_i\subseteq\cl{T}_i$ with cardinality $g_i+f$ satisfies condition~\eqref{eq:robust_ithrank}.
Therefore,
we rewrite problem~\eqref{eq:robust_minimal_prob} as follows:
\begin{argmini}
	{\substack{\cl{S}\subseteq[m]}}
	{|\cl{S}|}
	{\label{eq:robust_minimal_prob1}}
	{\cl{S}_\textnormal{r}^*\in}
	\addConstraint{|\cl{T}_i\cap\cl{S}|}{\ge g_i+f}{\;\forall i\in[p].}
\end{argmini}
We need not include subset $\cl{F}\subseteq\cl{S}$ in~\eqref{eq:robust_minimal_prob1} as the worst-case scenario with $f$ faulty actuators is handled by the redundant selection of $g_i+f$ indices from all subsets $\cl{T}_i$.
The rest of the proof follows analogously to the one in \Cref{app:setmulticover}.
	
	\bibliographystyle{IEEEtran}
    \bibliography{Bibfiles/SparseActuatorScheduling,Bibfiles/SparseActuatorScheduling-Luca,Bibfiles/SparseActuatorScheduling-intro}

@article{Tzoumas2016tcns-minimalActuatorPlacement,
	title = {Minimal Actuator Placement With Bounds on Control Effort},
	shorttitle = {{{minimalActuatorPlacement}}},
	author = {Tzoumas, V. and Rahimian, M. A. and Pappas, G. J. and Jadbabaie, A.},
	year = 2016,
	journal = {IEEE Trans. Contr. Netw. Syst.},
	volume = {3},
	number = {1},
	pages = {67--78},
	issn = {2325-5870},
	doi = {10.1109/TCNS.2015.2444031}
}

@inproceedings{Taha2017acc-actuatorSelectionCPS,
	title = {Actuator Selection for Cyber-Physical Systems},
	shorttitle = {{{actuatorSelectionCPS}}},
	booktitle = {Proc. {{American Control Conf}}.},
	author = {Taha, Ahmad F. and Gatsis, Nikolaos and Summers, Tyler and Nugroho, Sebastian},
	year = 2017,
	pages = {5300--5305},
	issn = {2378-5861},
	doi = {10.23919/ACC.2017.7963778}
}

@article{alexeev2012full,
  title={Full spark frames},
  author={Alexeev, Boris and Cahill, Jameson and Mixon, Dustin G},
  journal={J. Fourier Anal. Appl.},
  volume={18},
  number={6},
  pages={1167--1194},
  year={2012},
  publisher={Springer}
}

@ARTICLE{Manohar22tac-OptimalSensor,
  author={Manohar, Krithika and Kutz, J. Nathan and Brunton, Steven L.},
  journal= {IEEE Trans. Automat. Control}, 
  title={Optimal Sensor and Actuator Selection Using Balanced Model Reduction}, 
  year={2022},
  volume={67},
  number={4},
  pages={2108-2115},
  doi={10.1109/TAC.2021.3082502}}

@article{Argha19ijc-optimalActuatorSensorSelection,
  title = {A Framework for Optimal Actuator/Sensor Selection in a Control System},
  shorttitle = {{{optimalActuatorSensorSelection}}},
  author = {Argha, Ahmadreza and Su, Steven W. and Savkin, Andrey and Celler, Branko},
  year = {2019},
  journal = {Int. J. Control},
  volume = {92},
  number = {2},
  pages = {242--260},
  issn = {0020-7179},
  doi = {10.1080/00207179.2017.1350755},
  urldate = {2025-10-08}
}

@article{Baggio22ar-energyAwareControllability,
  title = {Energy-Aware Controllability of Complex Networks},
  shorttitle = {{{energyAwareControllability}}},
  author = {Baggio, Giacomo and Pasqualetti, Fabio and Zampieri, Sandro},
  year = {2022},
  month = may,
  journal = {Annu. Rev. Control Robot. Auton. Syst.},
  volume = {5},
  number = {1},
  pages = {465--489},
  issn = {2573-5144, 2573-5144},
  doi = {10.1146/annurev-control-042920-014957},
  urldate = {2024-04-10},
  langid = {english}
}

@article{Ballotta24lcss-actuatorSchedulingGuarantee,
  title = {Pointwise-Sparse Actuator Scheduling for Linear Systems With Controllability Guarantee},
  shorttitle = {{{actuatorSchedulingGuarantee}}},
  author = {Ballotta, Luca and Joseph, Geethu and Rahul Thete, Irawati},
  year = {2024},
  journal = {IEEE Control Syst. Lett.},
  volume = {8},
  pages = {2361--2366},
  issn = {2475-1456},
  doi = {10.1109/LCSYS.2024.3475886},
  urldate = {2025-01-24},
  copyright = {All rights reserved}
}

@article{Bjorklund09sjc-setPartitioning,
  title = {Set Partitioning via Inclusion-Exclusion},
  shorttitle = {{{setPartitioning}}},
  author = {Bj{\"o}rklund, Andreas and Husfeldt, Thore and Koivisto, Mikko},
  year = {2009},
  journal = {SIAM J. Comput.},
  volume = {39},
  number = {2},
  pages = {546--563},
  issn = {0097-5397},
  doi = {10.1137/070683933},
  urldate = {2025-10-07}
}

@inproceedings{Clark12cdc-leaderSelection,
  title = {On Leader Selection for Performance and Controllability in Multi-Agent Systems},
  shorttitle = {{{leaderSelection}}},
  booktitle = {Proc. {{IEEE Conf}}. {{Decis}}. {{Control}}},
  author = {Clark, Andrew and Bushnell, Linda and Poovendran, Radha},
  year = {2012},
  pages = {86--93},
  doi = {10.1109/CDC.2012.6426973},
  urldate = {2024-11-09}
}

@article{DeLorenzo90jgcd-SensorActuatorSelection,
  title = {Sensor and Actuator Selection for Large Space Structure Control},
  shorttitle = {{{SensorActuatorSelection}}},
  author = {DeLorenzo, M. L.},
  year = {1990},
  journal = {J. Guid. Control Dyn.},
  volume = {13},
  number = {2},
  pages = {249--257},
  issn = {0731-5090},
  doi = {10.2514/3.20544},
  urldate = {2025-10-15}
}

@article{Doostmohammadian20sj-driverNodes,
  title = {Minimal Driver Nodes for Structural Controllability of Large-Scale Dynamical Systems: Node Classification},
  shorttitle = {{{driverNodes}}},
  author = {Doostmohammadian, Mohammadreza},
  year = {2020},
  journal = {IEEE Syst.. J.},
  volume = {14},
  number = {3},
  pages = {4209--4216},
  issn = {1937-9234},
  doi = {10.1109/JSYST.2019.2956501},
  urldate = {2024-11-14}
}

@inproceedings{Duh97stoc-approximationkSetCover,
  title = {Approximation of {$K$}-Set Cover by Semi-Local Optimization},
  shorttitle = {{{approximationkSetCover}}},
  booktitle = {Proc. {{Annu}}. {{ACM Symp}}. {{Theory Comput}}.},
  author = {Duh, Rong-chii and F{\"u}rer, Martin},
  year = {1997},
  pages = {256--264},
  doi = {10.1145/258533.258599},
  urldate = {2025-10-07},
  isbn = {978-0-89791-888-6}
}

@inproceedings{Fujito06springer-betterThanGreedySetMulticover,
  title = {A Better-Than-Greedy Algorithm for $k$-Set Multicover},
  shorttitle = {{{betterThanGreedySetMulticover}}},
  booktitle = {Proc. Int. Workshop Approx. Online Algorithms},
  author = {Fujito, Toshihiro and Kurahashi, Hidekazu},
  editor = {Erlebach, Thomas and Persinao, Giuseppe},
  year = {2006},
  pages = {176--189},
  doi = {10.1007/11671411_14},
  isbn = {978-3-540-32208-5},
  langid = {english}
}

@article{Guo21tac-actuatorPlacementGreedy,
  title = {Actuator Placement Under Structural Controllability Using Forward and Reverse Greedy Algorithms},
  shorttitle = {{{actuatorPlacementGreedy}}},
  author = {Guo, Baiwei and Karaca, Orcun and Summers, Tyler and Kamgarpour, Maryam},
  year = {2021},
  journal = {IEEE Trans. Automat. Control},
  volume = {66},
  number = {12},
  pages = {5845--5860},
  issn = {1558-2523},
  doi = {10.1109/TAC.2020.3044284},
  urldate = {2024-11-07}
}

@article{Hua10tcs-dynamicProgrammingSetMulticover,
  title = {Dynamic Programming Based Algorithms for Set Multicover and Multiset Multicover Problems},
  shorttitle = {{{dynamicProgrammingSetMulticover}}},
  author = {Hua, Qiang-Sheng and Wang, Yuexuan and Yu, Dongxiao and Lau, Francis C. M.},
  year = {2010},
  journal = {Theor. Comput. Sci.},
  volume = {411},
  number = {26},
  pages = {2467--2474},
  issn = {0304-3975},
  doi = {10.1016/j.tcs.2010.02.016},
  urldate = {2025-10-15}
}

@article{Joseph24now-sparseActuatorControl,
  title = {Sparse Actuator Control of Discrete-Time Linear Dynamical Systems},
  shorttitle = {{{sparseActuatorControl}}},
  author = {Joseph, Geethu},
  year = {2024},
  month = sep,
  journal = {Found. Trends Syst. Control},
  volume = {11},
  number = {3},
  pages = {186--284},
  issn = {2325-6818, 2325-6826},
  doi = {10.1561/2600000033},
  urldate = {2025-10-15},
  langid = {english}
}

@article{Kolliopoulos05jcss-approximationAlgorithmsIntegerPrograms,
  title = {Approximation Algorithms for Covering/Packing Integer Programs},
  shorttitle = {{{approximationAlgorithmsIntegerPrograms}}},
  author = {Kolliopoulos, Stavros G. and Young, Neal E.},
  year = {2005},
  journal = {J. Computer Syst. Sci.},
  volume = {71},
  number = {4},
  pages = {495--505},
  issn = {0022-0000},
  doi = {10.1016/j.jcss.2005.05.002},
  urldate = {2025-10-07}
}

@article{Liu11nature-controllabilityNetworks,
  title = {Controllability of Complex Networks},
  shorttitle = {{{controllabilityNetworks}}},
  author = {Liu, Yang-Yu and Slotine, Jean-Jacques and Barab{\'a}si, Albert-L{\'a}szl{\'o}},
  year = {2011},
  journal = {Nature},
  volume = {473},
  number = {7346},
  pages = {167--173},
  issn = {1476-4687},
  doi = {10.1038/nature10011},
  urldate = {2025-10-15},
  copyright = {2011 Springer Nature Limited},
  langid = {english}
}

@article{Olshevsky14tcns-minimalControllability,
  title = {Minimal Controllability Problems},
  shorttitle = {{{minimalControllability}}},
  author = {Olshevsky, Alex},
  year = {2014},
  journal = {IEEE Trans. Control Netw. Syst.},
  volume = {1},
  number = {3},
  pages = {249--258},
  issn = {2325-5870},
  doi = {10.1109/TCNS.2014.2337974},
  urldate = {2024-03-08},
  langid = {english}
}

@article{Pasqualetti14tcns-controllabilityMetricsNetworks,
  title = {Controllability Metrics, Limitations and Algorithms for Complex Networks},
  shorttitle = {{{controllabilityMetricsNetworks}}},
  author = {Pasqualetti, Fabio and Zampieri, Sandro and Bullo, Francesco},
  year = {2014},
  journal = {IEEE Trans. Control Netw. Syst.},
  volume = {1},
  number = {1},
  pages = {40--52},
  issn = {2325-5870},
  doi = {10.1109/TCNS.2014.2310254},
  urldate = {2024-07-18}
}

@article{Pequito17automatica-robustMinimalControllability,
  title = {The Robust Minimal Controllability Problem},
  shorttitle = {{{robustMinimalControllability}}},
  author = {Pequito, S{\'e}rgio and Ramos, Guilherme and Kar, Soummya and Aguiar, A. Pedro and Ramos, Jaime},
  year = {2017},
  journal = {Automatica},
  volume = {82},
  pages = {261--268},
  issn = {0005-1098},
  doi = {10.1016/j.automatica.2017.04.053},
  urldate = {2025-08-07}
}

@article{Ramos21ijrnc-robustMinimalControllabilityObservability,
  title = {The Robust Minimal Controllability and Observability Problem},
  shorttitle = {{{robustMinimalControllabilityObservability}}},
  author = {Ramos, Guilherme and Silvestre, Daniel and Silvestre, Carlos},
  year = {2021},
  journal = {Int. J. Robust Nonlin. Control},
  volume = {31},
  number = {10},
  pages = {5033--5044},
  issn = {1099-1239},
  doi = {10.1002/rnc.5527},
  urldate = {2025-10-09},
  copyright = {{\copyright} 2021 John Wiley \& Sons Ltd.},
  langid = {english}
}

@article{Siami21tac-actuatorScheduling,
  title = {Deterministic and Randomized Actuator Scheduling With Guaranteed Performance Bounds},
  shorttitle = {{{actuatorScheduling}}},
  author = {Siami, Milad and Olshevsky, Alexander and Jadbabaie, Ali},
  year = {2021},
  journal = {IEEE Trans. Automat. Control},
  volume = {66},
  number = {4},
  pages = {1686--1701},
  issn = {1558-2523},
  doi = {10.1109/TAC.2020.3000976},
  urldate = {2024-03-06}
}

@inproceedings{Slavik96stc-greedySetCoverAnalysis,
  title = {A Tight Analysis of the Greedy Algorithm for Set Cover},
  shorttitle = {{{greedySetCoverAnalysis}}},
  booktitle = {Proc. {{Annu}}. {{ACM Symp}}. {{Theory Comput}}.},
  author = {Slav{\'i}k, Petr},
  year = {1996},
  pages = {435--441},
  doi = {10.1145/237814.237991},
  urldate = {2025-04-01},
  isbn = {978-0-89791-785-8}
}

@article{Summers16tcns-submodularity,
  title = {On Submodularity and Controllability in Complex Dynamical Networks},
  shorttitle = {Submodularity},
  author = {Summers, Tyler H. and Cortesi, Fabrizio L. and Lygeros, John},
  year = {2016},
  month = mar,
  journal = {IEEE Control Netw. Syst.},
  volume = {3},
  number = {1},
  pages = {91--101},
  issn = {2325-5870},
  doi = {10.1109/TCNS.2015.2453711},
  urldate = {2024-04-10}
}

@article{Tzoumas21tac-lqgCodesign,
  title = {{LQG} Control and Sensing Co-Design},
  shorttitle = {{{lqgCodesign}}},
  author = {Tzoumas, Vasileios and Carlone, Luca and Pappas, George J. and Jadbabaie, Ali},
  year = {2021},
  journal = {IEEE Trans. Autom. Control},
  volume = {66},
  number = {4},
  pages = {1468--1483},
  issn = {1558-2523},
  doi = {10.1109/TAC.2020.2997661}
}

@article{VanDeWal01automatica-selectionIO,
  title = {A Review of Methods for Input/Output Selection},
  shorttitle = {{{selectionIO}}},
  author = {{van de Wal}, Marc and {de Jager}, Bram},
  year = {2001},
  journal = {Automatica},
  volume = {37},
  number = {4},
  pages = {487--510},
  issn = {0005-1098},
  doi = {10.1016/S0005-1098(00)00181-3},
  urldate = {2025-10-09}
}

@article{Zare20tac-proximalAlgorithms,
  title = {Proximal Algorithms for Large-Scale Statistical Modeling and Sensor/Actuator Selection},
  shorttitle = {{{proximalAlgorithms}}},
  author = {Zare, Armin and Mohammadi, Hesameddin and Dhingra, Neil K. and Georgiou, Tryphon T. and Jovanovi{\'c}, Mihailo R.},
  year = {2020},
  journal = {IEEE Trans. Automat. Control},
  volume = {65},
  number = {8},
  pages = {3441--3456},
  issn = {1558-2523},
  doi = {10.1109/TAC.2019.2948268},
  urldate = {2024-07-04}
}

@article{Zhang23tac-observabilityRobustnessSensorFailures,
  title = {Observability Robustness Under Sensor Failures: A Computational Perspective},
  shorttitle = {{{observabilityRobustnessSensorFailures}}},
  author = {Zhang, Yuan and Xia, Yuanqing and Liu, Kun},
  year = {2023},
  journal = {IEEE Trans. Autom. Control},
  volume = {68},
  number = {12},
  pages = {8279--8286},
  issn = {1558-2523},
  doi = {10.1109/TAC.2023.3295698},
  urldate = {2025-10-08}
}

@article{Incer25-Pacti,
  title = {Pacti: Assume-Guarantee Contracts for Efficient Compositional Analysis and Design},
  shorttitle = {Pacti},
  author = {Incer, Inigo and Badithela, Apurva and Graebener, Josefine B. and Mallozzi, Piergiuseppe and Pandey, Ayush and Rouquette, Nicolas and Yu, Sheng-Jung and Benveniste, Albert and Caillaud, Benoit and Murray, Richard M. and {Sangiovanni-Vincentelli}, Alberto and Seshia, Sanjit A.},
  year = {2025},
  month = jan,
  journal = {ACM Trans. Cyber-Phys. Syst.},
  volume = {9},
  number = {1},
  pages = {3:1--3:35},
  issn = {2378-962X},
  doi = {10.1145/3704736},
  urldate = {2025-10-15}
}

@inproceedings{Zardini22cdc-modularCodesign,
  title = {Task-Driven Modular Co-design of Vehicle Control Systems},
  shorttitle = {{{modularCodesign}}},
  booktitle = {Proc. {{IEEE Conf}}. {{Decis}}. {{Control}}},
  author = {Zardini, Gioele and Suter, Zelio and Censi, Andrea and Frazzoli, Emilio},
  year = {2022},
  pages = {2196--2203},
  doi = {10.1109/CDC51059.2022.9993107},
  urldate = {2025-10-15}
}
	
\end{document}